\documentclass[pdflatex,sn-mathphys-num]{sn-jnl}

\usepackage{graphicx}%
\usepackage{multirow}%
\usepackage{amsmath,amssymb,amsfonts}%
\usepackage{amsthm}%
\usepackage{mathrsfs}%
\usepackage[title]{appendix}%
\usepackage{xcolor}%
\usepackage{textcomp}%
\usepackage{manyfoot}%
\usepackage{booktabs}%
\usepackage{algorithm}%
\usepackage{algorithmicx}%
\usepackage{algpseudocode}%
\usepackage{listings}%

\usepackage{amsfonts}
\usepackage{amsthm}
\usepackage{verbatim}
\usepackage{amsmath, amssymb}
\usepackage{tikz}
\usetikzlibrary{matrix, arrows}
\usepackage{listings}
\usepackage{color}
\usepackage{listings}
\usepackage[all]{xy}
\usepackage{bigints}
\usepackage{dsfont}

\theoremstyle{thmstyleone}%
\theoremstyle{thmstyletwo}%
\newtheorem{remark}{Remark}%

\theoremstyle{thmstylethree}%

\title{Pure or Unstable: A Generic Dichotomy for Strong Stackelberg Commitments}

\begin{document}


\newcommand{\cA}{\mathcal{A}}
\newcommand{\cB}{\mathcal{B}}
\newcommand{\cC}{\mathcal{C}}
\newcommand{\cD}{\mathcal{D}}
\newcommand{\cE}{\mathcal{E}}
\newcommand{\cF}{\mathcal{F}}
\newcommand{\cG}{\mathcal{G}}
\newcommand{\cH}{\mathcal{H}}
\newcommand{\cI}{\mathcal{I}}
\newcommand{\cJ}{\mathcal{J}}
\newcommand{\cK}{\mathcal{K}}
\newcommand{\cL}{\mathcal{L}}
\newcommand{\cM}{\mathcal{M}}
\newcommand{\cN}{\mathcal{N}}
\newcommand{\cO}{\mathcal{O}}
\newcommand{\cP}{\mathcal{P}}
\newcommand{\cQ}{\mathcal{Q}}
\newcommand{\cR}{\mathcal{R}}
\newcommand{\cS}{\mathcal{S}}
\newcommand{\cT}{\mathcal{T}}
\newcommand{\cU}{\mathcal{U}}
\newcommand{\cV}{\mathcal{V}}
\newcommand{\cW}{\mathcal{W}}
\newcommand{\cX}{\mathcal{X}}
\newcommand{\cY}{\mathcal{Y}}
\newcommand{\cZ}{\mathcal{Z}}
\newcommand{\bA}{\mathbb{A}}
\newcommand{\bB}{\mathbb{B}}
\newcommand{\bC}{\mathbb{C}}
\newcommand{\bD}{\mathbb{D}}
\newcommand{\bE}{\mathbb{E}}
\newcommand{\bF}{\mathbb{F}}
\newcommand{\bG}{\mathbb{G}}
\newcommand{\bH}{\mathbb{H}}
\newcommand{\bI}{\mathbb{I}}
\newcommand{\bJ}{\mathbb{J}}
\newcommand{\bK}{\mathbb{K}}
\newcommand{\bL}{\mathbb{L}}
\newcommand{\bM}{\mathbb{M}}
\newcommand{\bN}{\mathbb{N}}
\newcommand{\bO}{\mathbb{O}}
\newcommand{\bP}{\mathbb{P}}
\newcommand{\bQ}{\mathbb{Q}}
\newcommand{\bR}{\mathbb{R}}
\newcommand{\bS}{\mathbb{S}}
\newcommand{\bT}{\mathbb{T}}
\newcommand{\bU}{\mathbb{U}}
\newcommand{\bV}{\mathbb{V}}
\newcommand{\bW}{\mathbb{W}}
\newcommand{\bX}{\mathbb{X}}
\newcommand{\bY}{\mathbb{Y}}
\newcommand{\bZ}{\mathbb{Z}}

\newcounter{dummy} \numberwithin{dummy}{section}

\theoremstyle{definition}
\newtheorem{mydef}[dummy]{Definition}
\newtheorem{prop}[dummy]{Proposition}
\newtheorem{corol}[dummy]{Corollary}
\newtheorem{thm}[dummy]{Theorem}
\newtheorem{lemma}[dummy]{Lemma}
\newtheorem{eg}[dummy]{Example}
\newtheorem{notation}[dummy]{Notation}
\newtheorem{claim}[dummy]{Claim}
\newtheorem{Exercise}[dummy]{Exercise}
\newtheorem{question}[dummy]{Question}
\newtheorem{conjecture}[dummy]{Conjecture}

\newcommand{\bcr}{\begin{color}{red}}
\newcommand{\ec}{\end{color}}

\author[1]{\fnm{Kamil} \sur{Bulinski}}\email{kamil.bulinski@adelaide.edu.au}

\author[1]{\fnm{Lang} \sur{White}}\email{lang.white@adelaide.edu.au}

\author*[1]{\fnm{Hung} \sur{Nguyen}}\email{hung.nguyen@adelaide.edu.au}

\affil*[1]{\orgdiv{School of Computer Science and Information Technolgy}, \orgname{Adelaide University}, \orgaddress{\street{Street}, \city{City}, \postcode{5005}, \state{South Australia}, \country{Australia}}}




\abstract{We study the robustness of the Strong Stackelberg Equilibrium (SSE) in finite leader--follower games when the follower's best-response correspondence is set-valued. While optimistic tie-breaking (in the leader's favor) is commonly adopted, it can hinge on knife-edge indifferences. We formalize a stability notion: an SSE is \emph{unstable} if, at the leader’s committed strategy, the follower has an alternative best response that strictly reduces the leader’s payoff. Our main results establish a sharp generic dichotomy. Fixing the follower’s utility and sampling the leader’s utility from any continuous distribution, with probability one the optimal Stackelberg commitment is unique and is either (i) pure, or (ii) mixed and unstable. When both players’ utilities are sampled generically, this strengthens to: with probability one, the unique optimal commitment is either pure and stable or mixed and unstable. These theorems complement the classic generic-value result of von Stengel and Zamir by showing that even when optimistic and pessimistic leader values coincide generically, the \emph{strategy-level} SSE prediction is generically fragile whenever optimality requires genuine randomization. We further apply this perspective to Stackelberg satisfaction games, disproving a conjecture from prior work via counterexamples and identifying conditions under which it nonetheless holds.}


\maketitle

\section{Introduction}

Stackelberg (leader--follower) games model strategic environments in which one player (the \emph{leader}) can commit to a course of action before another player (the \emph{follower}) moves after observing that commitment. Originating in oligopoly theory and the study of first-mover advantage, the Stackelberg paradigm has long served as a canonical model of strategic commitment in economics and game theory; see, e.g., the classic treatment in \cite{FudenbergTirole1991} and modern discussions of leadership/commitment in strategic-form games \cite{StengelZamir2010}. In many settings, commitment can strictly benefit the leader relative to simultaneous-move play, and this benefit can be enhanced when the leader may commit to a mixed (randomized) strategy \cite{ConitzerSandholm2006,StengelZamir2010,FarinaMarchesiKroerGattiSandholm2018}.

In finite two-player Stackelberg games, once the leader commits to a mixed strategy, the follower selects a best response. A central modeling issue is that the follower's best-response correspondence may be set-valued, so the predicted leader payoff depends on how ties are resolved among follower best responses. Two standard conventions are the \emph{strong} (optimistic) and \emph{weak} (pessimistic) variants: in a \emph{Strong Stackelberg Equilibrium} (SSE) the follower breaks ties in the leader's favor, whereas in a \emph{Weak Stackelberg Equilibrium} (WSE) ties are broken against the leader \cite{FarinaMarchesiKroerGattiSandholm2018}. In algorithmic game theory, SSE has become the dominant solution concept, partly because it offers a clean optimistic benchmark and is amenable to computation in important classes of games \cite{ConitzerSandholm2006,StengelZamir2010}.

The optimistic tie-breaking convention is especially prevalent in \emph{Stackelberg security games}, where a defender commits to a randomized allocation or patrol schedule and an attacker best responds after observing (or learning) that strategy \cite{Tambe2011,AnTambeSSGOverview}. This modeling framework has motivated a large literature on algorithms and deployments, including systems for airport security, air marshal scheduling, port protection, and transit fare inspection \cite{Tambe2011,YinJiangJohnsonTambe2012TRUSTS}. At the same time, the tie-breaking assumption is not merely a technicality: in security games with additional feasibility constraints, the defender may be unable to induce a preferred best response via an infinitesimal perturbation, and optimistic SSE predictions may therefore be overly optimistic \cite{GuoEtAl2018InducibleAAMAS}. These considerations make it important to understand when and how optimistic equilibrium predictions are robust to indifference and small modeling perturbations.

A celebrated justification for optimistic tie-breaking is the generic-games analysis of von Stengel and Zamir \cite{StengelZamir2010}. In the mixed extension of finite bimatrix games, they show that for generic payoff matrices the leader's optimistic value (the \emph{Highest Leader Payoff}) coincides with the leader's pessimistic value (the \emph{Lowest Leader Payoff}), suggesting that optimistic tie-breaking does not change the attainable payoff frontier for ``most'' games \cite{StengelZamir2010}. However, equality of these benchmark values does not address the \emph{robustness} of the \emph{optimizing commitment} itself: an SSE recommendation may rest on knife-edge indifferences in the follower response correspondence, so that alternative best replies (still optimal for the follower) can yield sharply different leader payoffs.

Motivated by this gap, we study \emph{stability} properties of Strong Stackelberg equilibria. Informally, an SSE is \emph{unstable} if at the leader's committed strategy the follower has another best response that strictly decreases the leader's payoff. This stability notion is aligned with robustness concerns that have driven recent work incorporating bounded rationality, observation noise, or worst-case response deviations in Stackelberg models \cite{RobustSolutions2010,GanHanWuXu2025RSE}. Such robustness perspectives have been empirically and theoretically motivated in security settings where human followers may deviate from perfectly optimal best responses \cite{RobustSolutions2010,YangKiekintveldOrdonezTambeJohn2013HumanBehavior,CernyEtAl2021QuantalStackelbergEFG}.

\medskip
\noindent
\textbf{Main results.}
Our first main theorem (Theorem~\ref{thm: main thm random game}) establishes a sharp generic dichotomy. It states that if one fixes the follower utility function and samples the leader utility function from any continuous distribution then, with probability one, the optimal Stackelberg leader strategy is \emph{unique} and is either (i) \emph{pure}, or (ii) \emph{mixed and unstable}. Equivalently, if the game admits an optimal SSE commitment that requires genuine randomization, then generically this SSE is fragile: there exists an alternative follower best reply at that commitment that harms the leader.
We then show in Theorem~\ref{thm: u_L and u_F randomized} that when both players' utilities are sampled generically, this dichotomy strengthens: with probability one, the unique optimal commitment is either \emph{pure and stable} or \emph{mixed and unstable}.
These results complement the generic-value insight of \cite{StengelZamir2010}: while optimistic and pessimistic values may coincide generically, the \emph{strategy-level} prediction given by SSE is generically non-robust whenever it is genuinely mixed. Furthermore, Theorem~\ref{thm: quantitative bound} provides a quantitative version of our main results by showing that the gap in leader payoffs at an unstable impure SSE is non-negligible.

While our analysis focuses on finite normal-form Stackelberg games, related existence and computational issues for strong Stackelberg concepts in discounted stochastic games have been studied recently (e.g., stationary SSSE) \cite{BucareyEtAl2023SSSE}. More broadly, a growing literature studies \emph{robust} Stackelberg notions that model follower near-optimality or indifference directly, such as $\delta$-robust Stackelberg equilibrium and related formulations based on $\delta$-best responses~\cite{GanHanWuXu2025RSE}. Our generic instability theorem suggests that these robustifications are not only useful in constrained domains, but are generically relevant even in standard finite Stackelberg games: when the optimal optimistic commitment is mixed, robustness to small deviations in the follower’s response rule becomes structurally important.

Finally, recent learning and inference results further underscore that small deviations and bounded rationality are inevitable in deployed systems. For example, contextual/online formulations study Stackelberg learning with side information and regret guarantees~\cite{HarrisWuBalcan2024}, while inverse Stackelberg / inverse game-theoretic approaches highlight the role of bounded rationality (e.g., quantal response) in enabling learnability of follower preferences~\cite{WuShenFangXu2022}. Related work on no-regret learning in Stackelberg-like settings similarly treats deviations from exact best responses as intrinsic to the model~\cite{LiuRong2024,GoktasZhaoGreenwald2022}. In this context, our instability results provide theoretical backing for the pessimistic and robust modeling choices adopted in that literature: the optimistic SSE prescription is generically non-robust precisely when it relies on genuine randomization.

\medskip
\noindent
\textbf{Applications to satisfaction games.} We further apply our stability analysis to \emph{Stackelberg satisfaction games}, in which the follower's payoff is thresholded into a binary satisfaction indicator. Recent work introduced satisfaction and regret-based perspectives in Stackelberg games and conjectured that the leader's optimistic Stackelberg payoff is no worse in the Satisfaction game, with a proof for the case where the leader is restricted to pure commitments \cite{Satisfaction2024}. We show that the conjecture is false in general by providing explicit counterexamples, and we identify conditions under which it does hold. In particular, we show that whenever the original game admits an optimal \emph{pure} Stackelberg leader strategy, the conjecture holds; combined with our generic dichotomy, this yields an ``almost sure'' validation of the conjecture on the stable (typically pure) side of the generic regime.

We remark that since Stackelberg satisfaction games are Stackelberg games where the follower utility is $\{0,1\}$-valued, our Theorem~\ref{thm: main thm random game} applies (since only the leader utility has to be randomly sampled from a continuous distribution). Consequently, this implies that Robustness techniques of \cite{GanHanWuXu2025RSE} are essential and indeed applicable to Stackelberg satisfaction games. 

\textbf{Real-world applications.} Stackelberg security-game models have also supported operational decision-making in deployed systems for airport security (e.g.\ \textsc{ARMOR} at LAX) \cite{PitaEtAl2008AAAI_ARMOR,JainEtAl2010Interfaces}, air-marshal scheduling (e.g.\ \textsc{IRIS}) \cite{TsaiEtAl2009IRIS,JainEtAl2010Interfaces}, national-scale airport security resource allocation (e.g.\ \textsc{GUARDS}) \cite{PitaEtAl2011GUARDS_AAMAS}, maritime patrol planning for the U.S.\ Coast Guard (e.g.\ \textsc{PROTECT}) \cite{AnEtAl2012AIMag_PROTECT,AnEtAl2013INFORMS_PROTECT}, public-transit fare inspection (e.g.\ \textsc{TRUSTS}) \cite{YinEtAl2012TRUSTS}, wildlife protection patrol planning (e.g.\ \textsc{PAWS}) \cite{FangEtAl2016IAAI_PAWS} and in cybersecurity~\cite{ngo2024catch}. Recent optimization work also revisits the computation of SSE in large constrained security games via mixed-integer formulations and perfect-formulation approaches \cite{BucareyEtAl2023PerfectFormulations}.

For application domains such as security games, our results caution against overreliance on optimistic tie-breaking whenever the SSE-optimal recommendation is genuinely mixed: in that regime, the leader’s prescribed commitment is generically fragile, so small perturbations or alternative follower best replies can substantially reduce the leader’s realized payoff. 
In addition, our analysis of Satisfaction games sharpens and corrects the existing literature: we provide explicit counterexamples to the conjectured increase of leader payoff in the Satisfaction game, while also identifying a clean and interpretable sufficient condition—existence of a pure optimal Stackelberg commitment—under which the conjectured increase does hold. Since pure optimal commitments are generically prevalent on the stable side of our dichotomy, this condition offers a practically meaningful guideline for when the leader payoff is expected to not be worse in the Satisfaction game. This is significant because randomization of defender strategies is argued as a main motivation in security games.

\section{Background and statement of new results}

\subsection{Basic setup: Stackelberg Games and SSE} 

Let $S_L$ and $S_F$ denote the finite sets of pure strategies for the leader and follower respectively. Let $\Delta(S_L)$, $\Delta(S_F)$ be the probability measures on $S_L$ and $S_F$, i.e., the mixed strategies. Let $$u_L,u_F: S_L \times S_F \to \bR$$ be the utility functions of the leader and follower. As usual, we extend these utility functions to functions on $\Delta(S_L) \times \Delta(S_F)$ by defining
$$u(\sigma_L, \sigma_F) = \sum_{s_L \in S_L} \sum_{s_F \in S_F} \sigma_L(s_L) \, \sigma_F(s_F) \, u(s_L,s_F),\quad \text{for } \sigma_L \in \Delta(S_L), \sigma_F \in \Delta(S_F),$$

which is the expected utility when the leader plays a mixed (randomized) strategy $\sigma_L$, and the follower plays a mixed strategy $\sigma_F$ independently. 

Let
$$BR_{u_F}(\sigma_L) = \{ \sigma_F \in \Delta(S_F) ~|~ u_F(\sigma_L, \sigma_F) \geq u_F(\sigma_L, \sigma'_F) \quad \text{for all } \sigma'_F \in \Delta(S_F) \}$$ denote the Best Response set of the follower when the leader commits to the mixed strategy $\sigma_L \in \Delta(S_L)$.

\begin{remark} Throughout this note, we consider $S_L$ and $S_F$ fixed once and for all, however we will consider different utilities $u_F$ for the follower, hence we stress the dependence on the best response set on these.

\end{remark}

\begin{mydef}[Stackelberg Equilibrium]\label{def: SSE} Recall that a pair of mixed strategies $(\sigma^*_L, \sigma^*_F)$ is a \textbf{Strong Stackelberg equilibrium} (w.r.t $u_L, u_F$) if 
\begin{equation} \label{eqn: SSE} u_L(\sigma^*_L, \sigma^*_F) = \max_{\sigma_L \in \Delta(S_L)} \max_{\sigma_F \in BR_{u_F}(\sigma_L)} u_L(\sigma_L, \sigma_F). \end{equation}
We call such a $\sigma_L^* \in \Delta(S_L)$ an \textbf{optimal Stackelberg leader strategy}.
Note that the second maximum says that we break ties in favour of the leader.

\end{mydef}

\begin{mydef}\label{def: lowest leader payoff} The quantity $$\cL(u_L,u_F) = \sup_{\sigma_L \in \Delta(S_L)} \min_{\sigma_F \in BR_{u_F}(\sigma_L)} u_L(\sigma_L, \sigma_F)$$ has been studied in \cite{StengelZamir2010}, and it is known as the \textbf{Lowest Leader Payoff}. Similarly, the \textbf{Highest Leader Payoff} is defined to be \begin{equation} \label{eqn: highest leader payoff} \cH(u_L, u_F) = \max_{\sigma_L \in \Delta(S_L)} \max_{\sigma_F \in BR_{u_F}(\sigma_L)} u_L(\sigma_L, \sigma_F) \end{equation} 
which is the leader payoff at a Strong Stackelberg Equilibrium.

\end{mydef}

\subsection{Background: Tie-breaking assumptions and generic games}

In the literature it is often assumed that the Stackelberg follower breaks ties in favour of the leader \cite{Conitzer2016, Korzhyk2010, StengelZamir2010}. Note that this is implicit in the definition of the Strong Stackelberg Equilibrium given above. One may wonder whether this is a reasonable assumption. A quite strong justification for this assumption appears in the works of von Stengel and Zamir \cite{StengelZamir2010}, which states that for \textbf{generic games} the Highest Leader Payoff is equal to the Lowest Leader Payoff, where \textit{generic} means \textit{on a dense set of matrices}. To make this precise, observe that the set of all utilities functions $S_F \times S_L \to \bR$ can be identified with a Euclidean space $\bR^{S_L \times S_F} \cong \bR^{|S_L||S_F|}$, thus it makes sense to say that a set of utility functions is dense. The aforementioned theorem of von Stengel and Zamir may be stated as follows

\begin{thm}[Proposition 11 of \cite{StengelZamir2010}]\label{thm: StengelZamir H=L} There is a dense subset $\cU$ of $\bR^{S_L \times S_F} \times \bR^{S_L \times S_F}$ such that for all $(u_L, u_F) \in \cU$ we have  $$\cH(u_L, u_F) = 
\cL(u_L, u_F).$$ That is, $\cH(u_L, u_F) = \cL(u_L, u_F)$ holds for generic $u_L,u_F$.

\end{thm}

Note that $\cU$ can actually be taken to be a set of full Lebesgue measure\footnote{Although we have not been able to find a reference that explicitly says this, it is implicit as \cite{StengelZamir2010} mention that their notion of \textit{generic} coincides with the notion of \textit{generic} in \cite{Stengel2002}, where it is mentioned that \textit{generic games} form a set of full measure.}. Thus for most choices of utilities, we have $\cH(u_F, u_L) = \cL(u_F, u_L)$ and thus one may say that \textit{``it does not hurt''} to assume that the leader breaks ties in favour of the leader. We stress that the correct way of interpreting is as follows: for \textit{most games} even if the follower does not break ties in favour of the leader, then for each $\epsilon>0$ the leader can choose a mixed strategy $\sigma_{\epsilon} \in \Delta(S_L)$ such that $u_L(\sigma_{\epsilon}, \sigma_F)$ is $\epsilon$-almost as good as the optimal solution that one would get under the optimistic assumptions that ties are broken in favour of the leader, that is, 
$$u_L(\sigma_{\epsilon}, \sigma_F) \geq \cH(u_L, u_F) - \epsilon \quad \text{for all } \sigma_F \in BR_{u_F}(\sigma_{\epsilon}).$$

\subsection{Stability properties of generic Stackelberg games}

Even though for most games, breaking ties in favour of leader makes no difference to the \textit{global} maximum/supremum for the leader payoff in the sense that $\cH = \cL$, it might still make a difference as to what strategy the leader should choose. That is, it may well be the case that there is a difference at the Strong Stackelberg Equilibrium. This motivates the following definition.

\begin{mydef}\label{def: unstable} We say that a Strong Stackelberg equilibrium $(\sigma^*_L, \sigma^*_F)$ is \textbf{unstable} if there exists $\sigma'_F \in BR_{u_F}(\sigma^*_L)$ such that 
$$u_L(\sigma^*_L, \sigma'_F) <u_L(\sigma^*_L, \sigma^*_F) .$$ We say that it is $\textbf{stable}$ if it is not unstable. Likewise, we say that an optimal Stackelberg leader strategy $\sigma_L^*$ is \textbf{unstable} (resp. \textbf{stable}) if it appears as part of an \textbf{unstable} (resp. \textbf{stable}) Strong Stackelberg equilibrium. We call 
\begin{equation}\label{eqn: gap defn} \Delta(\sigma_L^*; S_L, S_F) = \max_{\sigma_F' \in BR_{u_F}(\sigma^*_L)} u_L(\sigma^*_L, \sigma^*_F) - u_L(\sigma^*_L, \sigma'_F) \end{equation} 
the \textbf{gap} of this Strong Stackelberg Equilibrium.

\end{mydef}

In other words, a Strong Stackelberg Equilibrium is unstable if the follower has another best reply which gives a worse payoff to the leader. Note that in zero-sum games (the case where $u_L = -u_F$) every Strong Stackelberg Equilibrium is stable.

Our main new result states that for a generic game, the Strong Stackelberg Equilibrium is either impure and unstable or it is pure. For us, generic means Lebesgue almost surely (which is stronger than \text{on a dense set}). The space of utility functions $S_L \times S_F \to \bR$ naturally has a Lebesgue measure as it can be identified with the Euclidean space $\bR^{|S_L||S_F|}$, which can naturally be identified with a space of matrices. Likewise, the reader can replace, without loss of generality, the Lebesgue measure with the uniform probability distribution on $[0,1]^{|S_L||S_F|}$ (so our theorem is about randomly chosen utilities $u_L$ for the leader). 

\begin{thm}\label{thm: main thm random game} Fix a utility $u_F:S_L \times S_F \to \bR$ for the follower. Then for Lebesgue almost all utilities of the leader $u_L:S_L \times S_F \to \bR$ the optimal Stackelberg leader strategy is unique and one of the following holds:

\begin{enumerate}
	\item The optimal Stackelberg leader strategy is pure.
	\item The optimal Stackelberg leader strategy is impure and unstable.
\end{enumerate}

\end{thm}

In other words, if the Stackelberg Game has a Strong Stackelberg Equilibrium that is not pure, then almost surely this equilibrium is unstable. On the other hand, if the Strong Stackelberg Equilibrium is pure, then we cannot say that it is almost surely stable or almost surely unstable (however, Theorem~\ref{thm: u_L and u_F randomized} below shows that we can say it is almost surely stable if we also sample $u_F$ randomly). Indeed, one can construct examples of $u_F$ where each case can occur on a set of positive measure.

Note that it immediately follows as a corollary that the Lebesgue measure can be replaced with the uniform probability distribution on $[0,1]^{|S_L||S_F|}$ or indeed any continuous distribution such as a Normal distribution. Thus an intuitive interpretation of this theorem says that if we randomly chosen utilities $u_L$ for the leader, with fixed $u_F$, then with probability $1$ it holds that either the (unique) Strong Stackelberg Equilibrium is pure or the Strong Stackelberg Equilibrium is non pure and unstable. We can also choose both $u_F$ and $u_L$ independently and randomly and the conclusion of the theorem is still true, as can be seen by conditioning on $u_F$. Thus our theorem is stronger when it is stated with $u_F$ fixed. Another motivation for fixing $u_F$ is that we are also interested in Satisfaction games \cite{Satisfaction2024} where $u_F$ is $\{0,1\}$-valued (on pure strategies) and thus is not randomly sampled from a continuous distribution.

We now turn to the question of what happens if we allow $u_F$ to also be randomized. We show that we can in fact strengthen the conclusion in Theorem~\ref{thm: main thm random game} as follows.

\begin{thm}\label{thm: u_L and u_F randomized} For Lebesgue almost all $(u_L,u_F) \in \bR^{S_L \times S_F} \times \bR^{S_L \times S_F}$ pairs of utilities for the leader and follower, the optimal Stackelberg leader strategy is unique and one of the following holds:

\begin{enumerate}
	\item The optimal Stackelberg leader strategy is pure and stable.
	\item The optimal Stackelberg leader strategy is non-pure and unstable.
\end{enumerate}

\end{thm}

Note that the difference in the conclusion of Theorem~\ref{thm: main thm random game} and Theorem~\ref{thm: u_L and u_F randomized} is that now if the optimal Stackelberg leader strategy is pure then almost-surely it is stable (in Theorem~\ref{thm: main thm random game}, it could be almost-surely unstable for particular $u_F$).

The practical implication of this is that, for almost all games, if the optimal Stackelberg leader strategy is pure then it is stable, while if it is non-pure then it is unstable and the leader should consider other mixed strategies to commit to. Observe that by Theorem~\ref{thm: StengelZamir H=L},  for almost all games the leader could adjust their strategy to get a payoff close to $\cH(u_L,u_F)$ even in the pessimistic case where ties are broken against the leader. However if via this adjustment process the leader gets too close to an unstable Stackelberg optimal leader strategy, then a follower who is willing to or tends to (perhaps inadvertently due to imperfect decision making)  make small sacrifices to their best response could harm the leader by choosing a response that results in a payoff much worse than $\cH(u_L, u_F)$. Thus naturally the leader should consider \textit{robustness} techniques, such as those studied in \cite{GanHanWuXu2025RSE} or \cite{RobustSolutions2010}. In particular, in \cite{GanHanWuXu2025RSE} a notion of a \textit{Robust Stackelberg Equilibrium} is introduced where the leader is assumed to only choose a best response that is $\delta$-close to optimal and breaks ties amongst such that are \textit{against} the leader. This provides a natural limit to how much the leader can (in a generic game) choose strategies that are sufficiently close to an unstable Strong Stackelberg Equilibrium. These methods are not computable, i.e. they don't scale.

We note that Theorem~\ref{thm: u_L and u_F randomized} is an immediate consequence of Theorem~\ref{thm: main thm random game} and the following Proposition, which could be of independent interest.

\begin{prop}\label{prop: BR unique} Let $u_L$ be a fixed utility of the leader. Then for Lebesgue-almost all utilities $u_F$ of the follower the following holds: for any pure strategy $s_L \in S_L$ the Best Response set $BR_{u_F}(s_L)$ has only one element.

\end{prop}

\subsection{Quantifying the gap in unstable strategies}

One may be interested in understanding whether the gap (as defined by equation (\ref{eqn: gap defn}) in Definition~\ref{def: unstable}) in an unstable strategy is negligible or large. Our next result is a quantitative sharpening of Theorem~\ref{thm: main thm random game} which says that if we sample from a large range of matrices then this gap is negligible only with small probability.

\begin{thm}\label{thm: quantitative bound} Fix a strategy $u_F$ of the follower and sample $u_L$ from the uniform distirbution on $[-W,W]^{S_L \times S_F}$, where $W>0$. Then for $\Delta>0$ the probability that the unique Stackelberg equilibrium is impure and has a gap of at most $\Delta$ is bounded above by $$O\left(\frac{\Delta}{W}\right) \quad \text{as } W \to \infty.$$ Here the implicit big-$O$ constant depends only on the cardinalities $|S_L|$ and $|S_F|$.
\end{thm}

We remark that this does not say that the gap is large with high probability. The correct interpretation is that, with high probability, either the optimal Stackelberg leader strategy is pure or the gap is large. 

\section{Applications to Satisfaction games}

In \cite{Satisfaction2024} the concept of \textit{satisfaction} games were considered and some interesting questions were posed, which we can now address using our result Theorem~\ref{thm: main thm random game}.

For a fixed $U_f^{-} \in \bR$ we can define a new utility, which we call the $U_f^{-}$-\textbf{satisfactionization} of $u_F$, of the follower $\widetilde{u_F}:S_L \times S_F \to \bR$ given by $$\widetilde{u_F}(s_L,s_F) = \mathds{1}\{ u_F(s_L, s_F) \geq U_f^{-}\}.$$ That is, $\widetilde{u_F}(s_L, s_F) = 1$ if $u_F(s_L, s_F) \geq U_f^{-}$ and $\widetilde{u_F}(s_L, s_F)  = 0$ otherwise. This gives rise to a new Stackelberg game called the \textbf{Satisfaction Game}.

In \cite{Satisfaction2024} it was conjectured that the highest leader payoff in the Satisfaction game is at least as good as the highest leader payoff in the original game. This can be formulated as follows.

\begin{conjecture}[cf. Proposition 1 in \cite{Satisfaction2024}] \label{conj} Let $\widetilde{u_F}$ by the satisfactionization as defined above. Then $$\cH(u_L, u_F) \leq \cH(u_L, \widetilde{u_F}).$$

\end{conjecture}

It turns out that this conjecture is false in general. In Section~\ref{sec: counterexample to conj}, we construct counterexamples to this conjecture. For the remainder of this section, we turn to providing some cases where this conjecture is actually true.

It was proven in Theorem 2 of \cite{Satisfaction2024} that Conjecture~\ref{conj} does indeed hold if we assume the leader can only choose pure strategies (i.e., in equations (\ref{eqn: SSE}) and (\ref{eqn: highest leader payoff}) the first $\max$ is over $S_L$ instead of $\Delta(S_L)$). The following is Lemma is an immediate consequence of this, but we provide a simple self-contained proof for convenience.

\begin{lemma}\label{lemma: conj true for pure} Conjecture~\ref{conj} is true when there is an optimal Stackelberg strategy for the original game $(u_L, u_F)$ that is pure.
\end{lemma}

\begin{proof} Let $s_L^* \in S_L$ be such an optimal strategy. Thus $$u_L(s_L^*, s_F^*) = \cH(u_L, u_F)$$ for some $s_F^* \in S_F$ as the best response set always has a pure strategy. We now claim that $$BR_{u_F}(s_L^*) \subset BR_{\widetilde{u_F}}(s_L^*),$$ i.e., the best response set for the original game is a subset of the best response set of the satisfaction game. To see this, let $s_F \in BR_{u_F}(s_L^*)$ be any pure strategy in the best response set or the original game.  Thus for any $s_F' \in S_F$ we have 
$$u_F(s_L^*, s_F') \leq u_F(s_L^*, s_F)$$ 
and hence 
$$\widetilde{u_F}(s_L^*, s_F') \leq \widetilde{u_F}(s_L^*, s_F)$$ 

since the function $\mathds{1}_{[U_F^{-1}, \infty)}$ is weakly-increasing. This precisely means that $s_F \in BR_{\widetilde{u_F}}(s_L^*)$. This inclusion of BR sets now implies that 
\begin{align*} \cH(u_L, \widetilde{u_F}) &\geq \max_{s_F \in BR_{\widetilde{u_F}}} u_L(s_L^*, s_F) \\
& \geq \max_{s_F \in BR_{u_F}} u_L(s_L^*, s_F) \\
&= \cH(u_L, u_F)
\end{align*}

which completes the proof.
\end{proof}

Now we can combine Lemma~\ref{lemma: conj true for pure}  with Theorem~\ref{thm: main thm random game} to show that Conjecture~\ref{conj} is true for generic games where the optimal Stackelberg strategy for the original game is stable. This can be precisely formulated as follows.

\begin{corol} Let $u_F$ be a fixed utility of the follower and let $\widetilde{u}_F$ be a satisfactionization of $u_F$ (utility in a satisfaction game derived from $u_F$). Then for Lebesgue almost all $u_L$ we have that the following holds: If the strong Stackelberg Equilibrium is stable for the original game (with utilities $u_F$, $u_L$) then $$\cH(u_L, u_F) \leq \cH(u_L, \widetilde{u_F}).$$ That is, Conjecture~\ref{conj} is almost surely true for games with stable strong Stackelberg equilibrium.

\end{corol}

\begin{proof} By Theorem~\ref{thm: main thm random game}, the optimal Stackelberg leader strategy, being stable, is Lebesgue-almost surely pure. Thus Lemma~\ref{lemma: conj true for pure} applies. \end{proof}

\section{Some examples illustrating Theorem~\ref{thm: main thm random game}}

Before we delve into the proof of Theorem~\ref{thm: main thm random game} in its full generality, it is worthwhile to prove it for a particular simple example and discuss some of its features. 

\begin{eg} We consider the case where $|S_L| = |S_F|=2$, i.e., both the leader and follower have two pure strategies.

As the theorem considers $u_F$ fixed, let us fix the matrix 

$$u_F =  \begin{bmatrix}
    1 &  0 \\
    0 & 1 \\
\end{bmatrix}$$

where as usual the row player is the leader and the column player is the follower. The leader matrix is random (randomly sampled from the Lebesgue measure or any continuous distribution on the space of matrices) and thus we represent it as 

$$u_L = \begin{bmatrix}
    a &  b \\
    c & d \\
\end{bmatrix}$$

where $a,b,c,d \in \bR$ are random and independent continuous real random variables.

\end{eg}

We suppose that the leader chooses the first row with probability $p \in [0,1]$ and the second row with probability $1-p$. The follower then chooses the first column with probability $q \in [0,1]$ and the second column with probability $1-q$.

The utility of the follower, considered as a function of $q$ with $p$ fixed, is then
\begin{align*} u_F(q) = pq + (1-p)(1-q) 
&= (2p - 1) q +1 - p \end{align*}

Thus we obtain that 

\begin{itemize}
	\item If $p < 1/2$, then follower chooses $q = 0$ as a best reply.
	\item If $p > 1/2$, then follower chooses $q = 1$ as a best reply.
	\item If $p = 1/2$, then any $q \in [0,1]$ is a best reply for the follower.
\end{itemize}

The utility of the leader is $$u_L(p) = apq + bp(1-q) + c (1-p)q + d(1-p)(1-q).$$
Seeing the dependence of $p$ on $q$ established in the three cases above, we see that the leader payoff is 

$$u_L(p) = 
\begin{cases} 
(b-d)p + d & \text{if } 0 \leq p < \frac{1}{2} \\
(a-c)p + c & \text{if } \frac{1}{2} \leq p < 1 \\
\frac{1}{2} (aq + b(1-q) + cq + d(1-q)) & \text{if } p = \frac{1}{2}.
\end{cases}$$

In other words, the leader payoff is a piecewise linear function of $p$ and could be multi-valued at $p = 1/2$ (depends on $q$ there). A strong Stackelberg Equilibrium precisely corresponds to a $p$ which maximizes this function.  A maximum thus must occur at either $p=0,1$ (a pure strategy for the leader) or at $p=1/2$ (a non-pure strategy of the leader). If a maximum occurs at $p=1/2$, then this Strong Stackelberg Equilibrium is stable if and only if $$\frac{1}{2} (aq + b(1-q) + cq + d(1-q))$$ is constant as a function of $q$. This happens only if $a-b+c-d = 0$. But as $a,b,c,d$ are drawn from a continuous distribution, this happens with probability $0$, thus illustrating the conclusion of our theorem. We note that, in this example, if a Strong Stackelberg Equilibrium occurs at a pure strategy (at $p =0$ or $p=1$), then it is stable as the best reply is unique. This need not be the case, there are some choices of $u_F$ where some pure strategy happens to be unstable almost surely (i.e., for almost all choices of $u_L$).

\section{Proof of Theorem~\ref{thm: main thm random game} and Theorem~\ref{thm: quantitative bound}}
\label{sec: main proof}

We will use the notation of \cite{StengelZamir2010} which we now recall as follows. It will be convenient to represent the utilities as matrices and the pure strategies as sets of integers. Thus we let $S_L = \{1, \ldots, N\} $ and $S_F = \{1, \ldots, M\}$ be the pure strategies for the leader and follower respectively. Thus we have a fixed $N \times M$ matrix $B$ that represents the utility of the follower (i.e., $u_F(i,j) = B_{i,j}$). Likewise the leader utilities can be represented by an $N \times M$ matrix $A$, i.e., $A_{i,j} = u_L(i,j)$, however this matrix $A$ is not fixed but sampled from the Lebesgue measure.

We let $$X(j) = \{ \sigma_L \in \Delta(S_L) ~|~ j \in BR_{u_F}(\sigma_L) \}$$ be the set of all mixed strategies of the leader to which $j$ is a best reply. We note that $X(j)$ is a convex polytope and thus is a convex hull of a finite set of \textbf{extreme points}. We stress that $X(j)$ does not depend on $u_L$ (i.e., not on $A$), which we have not fixed (only fixed $u_F$). We recall that an \textbf{extreme point} of a convex set $X$ is a point $x$ such that whenever $x = tx_1 + (1-t)x_2$ for some $x_1,x_2 \in X$ and $0 < t <1$ then $x_1=x_2=x$ (i.e., $x$ is extreme in $X$ if it cannot be written as a convex combination of two other elements of $X$). 
We let $\cE_j$ denote the extreme points of $X(j)$ and we let 
$$\cE = \bigcup_{j=1}^M \cE_j$$
be the union of all of these extreme points. We note that $\cE$ is a finite set that depends only on $u_F$ and not on $u_L$ and thus is fixed (non-random).

\begin{lemma} \label{lemma: unique optimal} For Lebesgue-almost all $u_L$ it holds that an optimal Stackelberg leader strategy (w.r.t $u_L$, $u_F$) is unique.

\end{lemma}

\begin{proof} An optimal leader strategy occurs at an extreme point of some $X(j)$ and thus at some point in $\cE$. There could be other optimal leader strategies but they will be convex combinations of such extreme points. As $\cE$ is finite, it will be enough to show that if $\sigma_L \neq \sigma_L'$ are distinct mixed leader strategies and $j, j' \in \{1, \ldots, M\}$ are pure follower strategies, then for Lebesgue-almost all $u_L$ we have that $u_L(\sigma_L, j) \neq u_L(\sigma_L', j')$.

We first consider the case where $j = j'$. In this case, $$u_L(\sigma, j) = \sum_{i=1}^N \sigma^i A_{i,j} $$ for any $\sigma = (\sigma^1, \dots, \sigma^N) \in \Delta(S_L)$. Thus $$u_L(\sigma_L, j) - u_L(\sigma_L', j) = \sum_{i=1}^N c_i A_{i,j}  $$ where $c_i$ is the $i$-th component of $\sigma_L - \sigma'_L$. Now the set of all matrices $A$ such that this expression is zero has Lebesgue measure zero, since at least one $c_i$ is non-zero. This is because the $c_i$ is here considered fixed and $A$ is randomly sampled, thus $(A_{1,j}, \ldots, A_{N,j}) \in \bR^N$ is sampled from the Lebesgue measure.

Now we must consider the case where $j \neq j'$. Say $j' = k \neq j$ for convenience of notation.

A similar calculation as above gives $$u_L(\sigma_L, j) - u_L(\sigma_L', k) = \sum_{i=1}^N c_i' A_{i,j} - \sum_{i=1}^N c_i'' A_{i,k}  $$ where $c'_i$ and $c_i''$ are the $i$-th components of $u_L(\sigma_L, j)$ and $u_L(\sigma_L', k)$ respectively.

Note that in this sum, at least one $c_i'$ and $c_i''$ are non-zero and the $$(A_{1,j}, \ldots, A_{N,j}, A_{1,k}, \ldots, A_{N,k}) \in \bR^{2N}$$ is sampled from the Lebesgue measure on $\bR^{2N}$. Thus this this expression is zero only on a set of zero Lebesgue measure. This completes the proof. \end{proof}

We now give the main geometric tool for proving Theorem~\ref{thm: main thm random game}.

\begin{prop}\label{prop: geo prop} Suppose that $u_L$ is chosen such that there is a unique strong Stackelberg Equilibrium $(\sigma_L^*, \sigma_F^*)$. Suppose also that that $\sigma_L^*$ is non-pure. Then $\sigma_L^*$ does not have a unique best reply, i.e., there exists $j \neq j'$ such that $$\sigma_L^* \in X(j) \cap X(j').$$

\end{prop}

\begin{proof} The function $\sigma_L \mapsto u_L(\sigma_L, \sigma_F)$ is a linear function on each $X(j)$, and thus have maxima occurring on its extreme points. Thus $\sigma_L^*$ is an extreme point of some $X(j_0)$ and is in fact an extreme point of any $X(j)$ that contains $\sigma_L^*$. Note that this is where we use uniqueness, as otherwise $\sigma_L^*$ could be a convex combination of such. As $\sigma_L^*$ is not pure, there exists two \textbf{distinct} $u, v \in \Delta(S_L)$ such that $$\sigma_L^* = pu + (1-p)v \quad\text{for some } p \in (0,1).$$ 
Now let 
$$u_n = \left(1 - \frac{1}{n}\right) \sigma_L^* + \frac{1}{n} u$$ 
and let 
$$v_n = \left(1 - \frac{1}{n}\right) \sigma_L^* + \frac{1}{n} v.$$

Notice that since $u \neq v$ we have $u_n \neq v_n$ and $$\sigma_L^* = pu_n + (1-p)v_n.$$ Thus since $\sigma_L^*$ is an extreme point inside any $X(j)$, it must follow that $u_n \in X(j_n)$ and $v_n \in X(j_n ')$ with $j_n \neq j_n'$. Since there are only finitely many choices for $j$, we may pass to a subsequence $n_1 < n_2 < \ldots$ of positive integers such that $j = j_{n_i}$ and $j = j'_{n_i}$ are independent of $n_i$. Thus $j \neq j'$ and $u_n \in X(j), v_n \in X(j')$. But since $X(j)$ and $X(j')$ are closed and $u_n, v_n \to \sigma_L^*$ as $n \to \infty$, we must have that $\sigma_L^* \in X(j) \cap X(j')$. \end{proof}

We are now in a position to prove Theorem~\ref{thm: main thm random game}. First, we have that for Lebesgue almost all $u_L$ that the optimal Stackelberg leader strategy $\sigma_L^* \in \Delta(S_L)$ is unique. If it is pure, then we are done. If it is not pure, then we know by Proposition~\ref{prop: geo prop} that there exists $j \neq k$ such that $\sigma_L^* \in X(j) \cap X(k)$. We now show that $\sigma_L^*$ is almost surely unstable. This will follow similar techniques to the proof of Lemma~\ref{lemma: unique optimal}. By writing $\sigma_L^* = (c_1, \ldots, c_N) \in \Delta(S_L)$, we have that
\begin{equation} \label{eqn: difference of utilities} u_L(\sigma_L^*, j) - u_L(\sigma_L^*,k) =  \sum_{i=1}^N c_i A_{i,j} - \sum_{i=1}^N c_i A_{i,k}  \end{equation} 
with at least one $c_i$ non-zero. But the vector $$(A_{1,j}, \ldots, A_{N,j}, A_{1,k}, \ldots, A_{N,k}) \in \bR^{2N}$$ is sampled from the Lebesgue measure on $\bR^{2N}$. Thus for Lebesgue almost all $u_L$ (i.e., Lebesgue almost all $A$), this expression is non-zero if we consider $\sigma_L^*$ fixed and $j$ and $k$ fixed. The proof is now complete by noting that there are only finitely many choices for the triple $\sigma_L^*, j,k$ as $\sigma_L^*$ must be in $\cE$, which is finite. This completes the proof of Theorem~\ref{thm: main thm random game}.

\begin{proof}[Proof of Theorem~\ref{thm: quantitative bound}] The proof of this theorem is the same as the proof of Theorem~\ref{thm: main thm random game} above but with a different ending. Here we sample $u_L$ from $[-W,W]^{S_L \times S_F}$. We are now interested when the quantity from Equation (\ref{eqn: difference of utilities}) given by 
$$u_L(\sigma_L^*, j) - u_L(\sigma_L^*,k) =  \sum_{i=1}^N c_i A_{i,j} - \sum_{i=1}^N c_i A_{i,k} $$
has absolute value at most $\Delta$. Here note that $$(A_{1,j}, \ldots, A_{N,j}, A_{1,k}, \ldots, A_{N,k}) \in [-W,W]^{2N}$$ is sampled from the uniform distribution on $[-W,W]^{2N}$. Now for some $n \in \{1, \ldots, N\}$ we must have $c_n\geq 1/N$. Observe that if we fix all the variables except for the single variable $A_{n, j}$, then the probability that $A_{n,j} \in [-W,W]$ is chosen so that 
$$\sum_{i=1}^N c_i A_{i,j} - \sum_{i=1}^N c_i A_{i,k} \in [-\Delta, \Delta]$$ 
is at most $$ \frac{2\Delta c_n^{-1}}{2W} \leq \frac{N\Delta}{W}.$$

Thus by conditioning on these variables (i.e., all $2N-1$ variables except for $A_{n,j}$) we see that, with $(c_1, \ldots, c_N)$ fixed and $$(A_{1,j}, \ldots, A_{N,j}, A_{1,k}, \ldots, A_{N,k}) \in [-W,W]^{2N}$$ sampled uniformly,  the probability that
\begin{equation}\label{eqn: inside Delta interval} \sum_{i=1}^N c_i A_{i,j} - \sum_{i=1}^N c_i A_{i,k} \in [-\Delta, \Delta] \end{equation}
is at most $N\Delta/W$. Now there are only finitely many choices for $\sigma_L^* = (c_1, \ldots, c_N)$, as we recall that these are points inside the finite set $\cE$. We remark that this point $\sigma_L^*$ depends on the matrix $A$, however for the purposes of giving an upper bound on a probability we can consider it fixed and use the union bound as follows. By taking the union over all such $\sigma_L^* = (c_1, \ldots, c_N) \in \cE$ and all $j<k$ we see that (\ref{eqn: inside Delta interval}) holds for all such choices of $\sigma_L^*$, $j$,$k$ with probability 
$$\binom{N}{2}|\cE| \frac{N\Delta}{W} = O(W), $$ which completes the proof as $|\cE|$ is finite and can in fact be bounded by an expression involving only $N$ and $M$. 
\end{proof}

\section{Proof of Proposition~\ref{prop: BR unique} and Theorem~\ref{thm: u_L and u_F randomized}}

We now turn to proving Proposition~\ref{prop: BR unique}. We use the same notation as in Section~\ref{sec: main proof} regarding matrices and our indexing  $S_L = \{1, \ldots, N\} $ and $S_F = \{1, \ldots, M\}$. This time the $N \times M$ matrix $B$ that represents the utility of the follower is random while the $N \times M$ matrix $A$ that represents the utility of the leader is fixed. If $i, j \in S_F = \{1, \ldots, M\}$ are two different pure strategies of the leader that are both in $BR_{u_F}(k)$, for some $k \in S_L$, then we must have $$u_F(k, i) = u_F(k, j).$$ In terms of matrices this means that $$B_{k,i} = B_{k,j}.$$ However, since $B$ is sampled from the Lebesgue measure, the probability that this holds is $0$ as $i \neq j$.

We note that Proposition~\ref{prop: BR unique} and Theorem~\ref{thm: main thm random game} now imply Theorem~\ref{thm: u_L and u_F randomized}.

\section{Counterexamples to Conjecture~\ref{conj}}
\label{sec: counterexample to conj}
The purpose of this section is to show that this Conjecture~\ref{conj} is false by constructing a class of counterexamples. 

\begin{proof}[Disproof of Conjecture~\ref{conj}] Consider the game
$$G = \begin{bmatrix}
    a,1 &  b ,0 \\
    c,1 & d ,2 \\
\end{bmatrix}$$
where $a,b,c,d \in \bR$ are real numbers satisfying
\begin{align}
a < c < d < b
\end{align}
Then the satisfaction game with $U_f^{-} = 1$ has matrix 
$$\widetilde{G} = \begin{bmatrix}
    a,1 &  b ,0 \\
    c,1  & d ,1 \\
\end{bmatrix}.$$

Let us first compute the Strong Stackelberg Equilibrium (and hence Highest Leader Payoff) in the original game $G$.

We suppose that the leader chooses the first row with probability $p \in [0,1]$ and the second row with probability $1-p$. The follower then chooses the first column with probability $q \in [0,1]$ and the second column with probability $1-q$.

The utility of the follower (in the original game $G$) is then
\begin{align*} u_F(q) = pq + (1-p)q + 2(1-p)(1-q) 
&= (2p - 1) q - 2p + 2 \end{align*}

Thus we obtain that 

\begin{itemize}
	\item If $p < 1/2$, then follower chooses $q = 0$ as a best reply.
	\item If $p > 1/2$, then follower chooses $q = 1$ as a best reply.
	\item If $p = 1/2$, then any $q \in [0,1]$ is a best reply for the follower.
\end{itemize}

The utility of the leader is $$u_L(p) = apq + bp(1-q) + c (1-p)q + d(1-p)(1-q).$$

This implies that

\begin{itemize}
	\item If $p < 1/2$, then $q = 0$ and so $$u_L(p) = bp + d(1-p) = (b-d)p + d.$$ Thus since $d< b$ this function is increasing in $p$ and so we have 
$$u_L(p) < (b-d) \cdot \frac{1}{2} + d = \frac{1}{2}(b+d).$$
	\item If $p > 1/2$, then $q = 1$ and so 
$$u_L(p) = ap + c(1-p) = (a-c)p + c $$ and since $a<c$ this function is decreasing in $p$ and so we have  
$$u_L(p) < (a-c) \cdot \frac{1}{2} + c = \frac{1}{2}(a+c).$$
	\item If $p = 1/2$ then any $q \in [0,1]$ is a best reply and $$u_L(p) = \frac{1}{2}(a+c)q + \frac{1}{2}(b+d)(1-q).$$ As we assume the follower breaks ties in favour of the leader, the follower chooses $q = 0$ since $\frac{1}{2}(a+c) \leq \frac{1}{2}(b+d)$.
\end{itemize}

This means that the unique Strong Stackelberg Equilibrium in the original game $G$ occurs at $p = 1/2$ and $q = 0$ which gives a leader utility of $$\frac{1}{2}(b+d).$$

Now let us compute the Strong Stackelberg Equilibrium in the satisfaction game $\widetilde{G}$. The utility of the follower is 
$$\widetilde{u_F}(q) = pq + (1-p)q + (1-p)(1-q) = pq + 1- p.$$

This means that when $p>0$, the best reply of the follower is $q=1$, which gives a leader utility of $$u_L(p) = ap + c(1-p) \leq c.$$

If $p=0$ then any $q \in [0,1]$ is a best reply of the follower and the leader utility is $$u_L(p) = cq + d(1-q) \leq d.$$

Thus since $c<d$, the Highest Leader Payoff in the satisfaction game is $d$ and occurs at the Strong Stackelberg Equilibrium $p = 0, q = 0$. However 
$$d < \frac{1}{2}(b+d) $$ as $b > d$ and so in fact the Highest Leader Payoff in the original game is strictly higher. \end{proof}

\section{Conclusion}
This paper isolates a strategy-level fragility in the strong Stackelberg equilibrium (SSE) that is not captured by classic value-based justifications of optimistic tie-breaking: introducing a stability notion (an SSE is \emph{unstable} if the follower has another best response at the leader’s commitment that strictly lowers the leader’s payoff), we prove a sharp generic dichotomy showing that, for Lebesgue-almost all leader payoffs (with fixed follower payoffs), the optimal Stackelberg leader strategy is unique and is either \emph{pure} or \emph{mixed and unstable}; moreover, for Lebesgue-almost all payoff pairs it is either \emph{pure and stable} or \emph{mixed and unstable}. These results complement the von Stengel--Zamir generic-value insight (optimistic and pessimistic leader payoff benchmarks coincide generically) by demonstrating that whenever optimality genuinely requires randomization, the SSE recommendation is generically non-robust to indifference and alternative best replies. We then apply this perspective to Stackelberg satisfaction games, disproving the conjecture that satisfactionization can never reduce the leader’s optimistic value via explicit counterexamples, while also showing that this conjecture holds whenever the original game admits an optimal \emph{pure} Stackelberg leader strategy—yielding an ``almost sure'' validation on the stable (typically pure) side of the generic regime. Practically, our results recommend skepticism toward optimistic SSE solutions that require genuine randomization, and motivate selecting commitments via pessimistic/robust criteria whenever indifference cannot be reliably resolved.

\bibliography{ref.bib}

\end{document}